\documentclass[letterpaper, 10 pt, conference,twocolumn]{ieeeconf}  
\IEEEoverridecommandlockouts                              
\usepackage{graphics} 
\usepackage{graphicx,color}
\usepackage[active]{srcltx}
\usepackage{amsmath} 
\usepackage{amssymb}  
\usepackage{amsmath,amssymb,amsfonts,theorem}
\usepackage{enumerate}
\usepackage{tikz}
\usepackage{tabu}
\usepackage{algpseudocode} 

\usepackage{epstopdf}

{ \theorembodyfont{\normalfont} 

\DeclareMathOperator{\e}{\mathrm{e}} 

\newcommand{\norm}[1]{\left\lVert #1 \right\rVert_2}

\makeatletter
\renewcommand*{\d}%
{\@ifnextchar^{\DIfF}{\DIfF^{}}}
\def\DIfF^#1{%
	\mathop{\mathrm{\mathstrut d}}%
	\nolimits^{#1}\gobblespace
}
\def\gobblespace{%
	\futurelet\diffarg\opspace}
\def\opspace{%
	\let\DiffSpace\!%
	\ifx\diffarg(%
	\let\DiffSpace\relax
	\else
	\ifx\diffarg\[%
	\let\DiffSpace\relax
	\else
	\ifx\diffarg\{%
	\let\DiffSpace\relax
	\fi\fi\fi\DiffSpace
}

\newcommand{\R}{\mathbb{R}}

\renewcommand{\S}{\mathcal{S}}    

\newcommand\ie{\emph{i.e.}}

{
\newtheorem{example}{Example}

}
\newtheorem{assumption}{Assumption}

\newtheorem{lemma}{Lemma}
\newtheorem{corollary}{Corollary}
\newtheorem{proposition}{Proposition}

\def\QEDclosed{\mbox{\rule[0pt]{1.3ex}{1.3ex}}} 
\def\QEDopen{{\setlength{\fboxsep}{0pt}\setlength{\fboxrule}{0.2pt}\fbox{\rule[0pt]{0pt}{1.3ex}\rule[0pt]{1.3ex}{0pt}}}}
\def\QED{\QEDopen}

\def\be{\begin{equation}}
\def\ee{\end{equation}}
\def\ba{\begin{array}}
\def\ea{\end{array}}
\def\eqa{\begin{eqnarray}}
\def\eqe{\end{eqnarray}}

\title{\LARGE \bf Analysis of free recall dynamics of an abstract working memory model$^{\star}$
\author{Gianluca Villani, Matin Jafarian$^{\ast}$, Anders Lansner, Karl Henrik Johansson}
\thanks{$^{\star}$This work was supported by the Knut and Alice Wallenberg Foundation, the Swedish Strategic Research Foundation and the Swedish Research Council. The authors are with the School of Electrical Engineering and Computer Science, KTH Royal Institute of Technology, Stockholm, Sweden. 
$^{\ast}$Corresponding to matinj@kth.se.}}
\begin{document}
\maketitle
\thispagestyle{empty}
\pagestyle{empty}
\begin{abstract}
This paper analyzes the {\em free recall} dynamics of a working memory model. Free recalling is the reactivation of a stored pattern in the memory in the absence of the pattern. Our free recall model is based on an abstract model of a modular neural network composed on $N$ modules, {\em hypercolumns}, each of which is a bundle of {\em minicolumns}. This paper considers a network of $N$ modules, each consisting of two minicolumns, over a complete graph topology. We analyze the free recall dynamics assuming a constant, and homogeneous coupling between the network modules. We obtain a sufficient condition for synchronization of network's minicolumns whose activities are positively correlated. Furthermore, for the synchronized network, the bifurcation analysis of one module is presented. This analysis gives a necessary condition for having a stable limit cycle as the attractor of each module. The latter implies recalling a stored pattern. Numerical results are provided to verify the theoretical analysis. 
\end{abstract}
\section{Introduction}\label{sec:intro}
Working Memory (WM) is a general-purpose cognitive system responsible for temporary holding information in service of higher-order cognition such as reasoning and decision making. The importance of understanding human memory functioning is evident from its central role in our cognitive functions \cite{d2015cognitive} as well as its role as the main inspiration behind developments in artificial memory networks \cite{graves2016hybrid}. 

Among the most important features of working memory, is the spontaneous free recall process. The latter refers to reactivation of a stored pattern in the memory in the absence of the pattern. 
 The precise mechanisms underlying free recall dynamics in the human brain is not yet fully understood \cite{lansner2013reactivation}. Yet, several abstract neuro-computational \cite{fiebig2018active} as well as more detailed spiking neural network models have been built based on different neurobiological hypotheses \cite{lansner2009associative}, at different levels of abstractions, to account for human experimental data on working memory.

Attractor neural networks, dynamical networks which evolves towards a stable pattern, have been employed for understanding the mechanisms of human memory, including working memory. Among the most studied models is the Hopfiled model \cite{hopfield1982neural} (and its several variations e.g. \cite{gray2017agent}) which represents a memory network with a fully connected graph, symmetric weights and capable of storing only binary values. The model has limitations on the capacity as well as recalling previously stored patterns \cite{fiebig2018active}. The latter has motivated designs with more than two states including the modular recurrent networks of Potts type \cite{kanter1988potts}.

The model in this paper is originated from the biologically inspired modular model of Potts type in \cite{lansner2013reactivation}. The network model in \cite{lansner2013reactivation} is composed of $N$ modules, hypercolumns, each of which consists of a bundle of elementary units, minicolumns, interacting via lateral inhibition, such that each hypercolumn module acts as a
winner-take-all microcircuit. We refer the interested reader for a detailed biological rationale of this model to \cite{lansner2013reactivation}. 

In this paper, we study the free recall dynamics of WM based on a simplification of the model in \cite{lansner2013reactivation}. We consider a network of $N$ hypercolumns, each consisting of two minicolumns. We assume that a few patterns have been encoded in this WM by means of fast Hebbian plasticity \cite{fiebig2017spiking} using exogenous signals. In this paper, we study the post-training reactivation dynamics of such a multi-item WM assuming a constant and homogeneous coupling between the network modules. Using tools from stability theory (a Lyapunov-based argument, e.g. \cite{jafarian2018syn}), we obtain a sufficient condition for achieving synchronization of network units which are positively correlated assuming a complete graph topology. Furthermore, for the synchronized network, the bifurcation analysis of one module's dynamics is presented. This analysis gives a necessary condition for having a stable limit cycle as the attractor of each network module. The latter implies recalling a stored pattern.

To the best of our knowledge, the free recall dynamics of working memory has not been studied analytically before, in particular from a control theory perspective. As shown in this paper, such analysis is useful for a better understanding of the free recalling mechanism of working memory networks.
 
The paper is organized as follows. Section \ref{sec:problem}, presents the model and problem formulation. Section \ref{sec:analysis}, presents a Lyapunov analysis for characterizing synchronization condition. This section also provides the bifurcation analysis of one module of the network assuming that the network is synchronized. In Section \ref{sec:sim}, simulation results are presented. Section \ref{sec:conclude} concludes the paper.
\section{Problem formulation}\label{sec:problem}
In this section, we present our model and state our goal of analysis. We model the free-recall dynamics of WM based on a simplification of a (non-spiking) attractor neural network originated from a modular
recurrent neural network model \cite{lansner2013reactivation}. 
                                                                                                                                                                                                                                                               This network is composed of $N$ modules, hypercolumns, each of which consists of $m$ elementary units, $m$ minicolumns. 
In this model, each memory input (or pattern) is encoded by a mechanism into $N$ attributes (represented by hypercolumns) each of which is composed of $m$ intervals (represented by minicolumns).\\
Figure \ref{f1} shows a configuration of a network composed of $3$ hypercolumns. 
\begin{figure}
\centering
\includegraphics[scale=0.22]{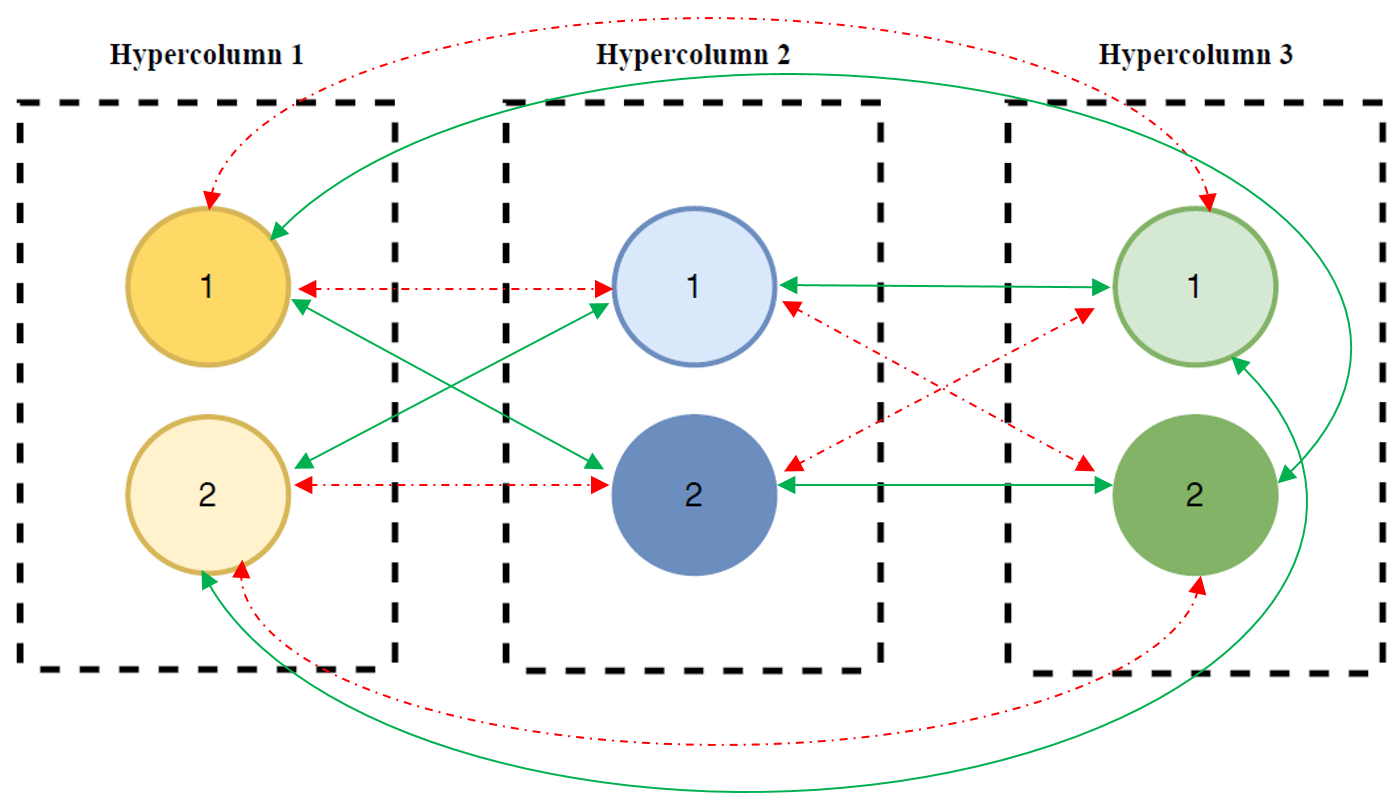}
\caption{\label{f1}A network composed of $3$ hypercolumns. Each hypercolumn is composed of 2 minicolumns. Each minicolumn is connected via positive (excitatory, shown by solid lines) coupling to the minimcolumns in other hypercolumns if their activities are correlated. Otherwise the coupling is negative (inhibitory, shown by dashed lines).}
\end{figure}
We assume that the memory has been trained to learn some patterns using exogenous signals. In this paper, we study the post-training dynamics. Our goal is analyzing the network's dynamics in order to obtain conditions under which a stored pattern is recalled.\\[1.5mm] 
We represent the network by a connected, undirected complete graph composed of $N$ nodes. Each node corresponds to one hypercolumn $i \in \{1,\ldots,N\}$ whose state is denoted by $s_i$. We study the case where $s_i \in \R^2$, i.e., $s_i=[s_{i1},s_{i2}]$, where $s_{ij}$ is the state of minicolumn $j$, $j \in \{1,2\}$, of hypercolumn $i$.\\[-0.4mm]

We model the dynamics of each minicolumn $s_{ij}$ based on a simplification of the model in \cite{lansner2013reactivation}, as follows  
\begin{equation}\label{s_dynaics_recall-1}
\dot s_{ij} = \sum\limits_{\substack{k=1 \\ k\neq i}}^N\sum\limits_{l=1}^2 {\omega}_{kl,ij} o_{kl} -a_{ij}-s_{ij}, 
\end{equation}
\begin{equation}\label{a_dynaics_recall}
\tau \dot a_{ij} = {g_a} o_{ij}-a_{ij}, \quad \tau >1,
\end{equation}
\begin{equation}\label{softmax_definition}
o_{ij} = \sigma(s_{ij})= \frac{\mathrm{e}^{s_{ij}}}{\sum\limits_{k=1}^m \mathrm{e}^{s_{ik}}},
\end{equation}
where $s_{ij} \in \R$, $a_{ij} \in \R$, and $o_{ij} \in (0,1)$ represent the level of activation, the level of the adaptation, and the output of the minicolumn $m_{ij}$ (minicolumn $j$ in the hypercolumn $i$), receptively. The parameters $g_a >0$ and $\tau >1$ are constant. 
The coupling wight of the interconnection of the two minicolumns $s_{kl}$ (in hypercolumn $k$) and $s_{ij}$ (in hypercolumn $i$) is represented by ${\omega}_{kl,ij}$. We assume that the coupling weights are constant in the post-training dynamics. Here, we apply Hebb's rule in the way that minicolumns are interconnected. This implies that the mini-columns which were activated simultaneously in the training phase are connected with positive couplings in the post-taining dynamics, while the minicolumns with uncorrelated activities are connected with negative couplings (see Fig \ref{f1}).

The role of variable $a_{ij}$, which models the biological mechanism of neural adaptation, is to deactivate its corresponding minicolumn $m_{ij}$ in response to the changes in the activity of the other minicolumn in hypercolumn $i$.  Owing to the definition of $o_{ij}$ in \eqref{softmax_definition}, units within the same hypercolumn interact via {\em lateral inhibition}, such that each hypercolumn acts as a soft winner-take-all microcircuit. This means that what determines the output of each minicolumn is the relative level of its activation with respect to the activation of the other minicolumn in the same hypercolumn. This interaction is modeled by a {\em soft-max} function, $\sigma$ as defined in Equation \eqref{softmax_definition}.

\begin{assumption}\label{ass2}
All interconnection weights are equal such that $|{\omega}_{kl,ij}|=\frac{\omega}{2}, \omega>0$.
\end{assumption}
Without loss of generality, we assume that any two minicolumns $i1$ and $k1$ ($i2$ and $k2$) are connected by positive coupling, whereas the interconnection of $i1$ to $k2$ is negative, $\forall i, k \in \{1,\ldots,N\}$. Since each hypercolumn is composed of two minicolumns, we can write $o_{i1}= 1- o_{i2}, \forall i$ based on the definition of the $\sigma$ function in \eqref{softmax_definition}.\\[0.5mm]
Now, with Assumption \ref{ass2}, the dynamics of each minicolumn in \eqref{s_dynaics_recall-1} can be written as follows
\begin{align}\label{s_dynaics_recall}
\dot s_{ij} &= \omega \sum\limits_{\substack{k=1 \\ k\neq i}}^N \sigma(s_{kj}) - a_{ij}-s_{ij}- (N-1) \frac{\omega}{2},
\end{align}
where $i \in \{1, \ldots, N\}, j \in \{1,2\}$. The hypercolumn dynamics (the node dynamics) obeys
\begin{equation}\label{s_dynaics_recall-n}
\dot s_{i} =  -s_{i}-a_{i} -\frac{\kappa}{2} + \omega \sum\limits_{\substack{k=1 \\ k\neq i}}^N  \sigma(s_k),\quad \omega>0,
\end{equation}
\begin{equation}\label{a_dynaics_recall-n}
\tau \dot a_{i} = {g_a} \sigma(s_i)-a_{i}, \quad \tau >1,
\end{equation}
where $s_i = [s_{i1},s_{i2}]^T$, $a_i = [a_{i1}, a_{i2}]^T$ , $\sigma(s_i)=[\sigma(s_{i1}),\sigma(s_{i2})]^T$, and $\kappa=(N-1) \omega$.\\[2mm]
In the next section, we analyze the network dynamics (the free-recall dynamic) with the node dynamics as in \eqref{s_dynaics_recall-n}, \eqref{a_dynaics_recall-n} and to answer the question that under which conditions a stored pattern is recalled. 
From a control theory perspective, this question is translated to characterizing coupling conditions under which the network modules synchronize \cite{sepulchre2006oscillators} and oscillate according to the desired pattern.
\section{Analysis}\label{sec:analysis}
In this section, we first derive the coupling condition under which synchronization, \ie\ $s_{i,1} = s_{k,1}; s_{i,2} = s_{k,2}; a_{i,1}=a_{k,1}; a_{i,2} = a_{k,2}$, can be achieved for the network with the dynamics in \eqref{s_dynaics_recall-n}, \eqref{a_dynaics_recall-n}. We then analyze the behavior of one module (one hypercolumn) in the synchronized network using tools from the bifurcation theory. The analysis gives a coupling condition for having a stable limit cycle behaviour for each hypercolumn in the synchronized network.
\begin{proposition}\label{pr1}
Consider the network dynamical system in \eqref{s_dynaics_recall-n}- \eqref{a_dynaics_recall-n}. If $\frac{g_a}{\tau} < \omega < \frac{g_a}{\tau-1}$ holds, the solution to the system in \eqref{s_dynaics_recall-n}- \eqref{a_dynaics_recall-n} converges to the set ${\S}=\{s_i \in \R, a_i \in \R: s_i = s_k; a_i=a_k, \forall i, k \in \{1,\ldots,N\}\}$.
\end{proposition}
\begin{proof}
Define $\alpha = \frac{1}{\tau},\ \alpha <1 $. Define $D_k= s_1-s_k$ and $E_k=a_1-a_k$. By multiplying the equation in \eqref{a_dynaics_recall-n} with $\alpha$, the error dynamics for any two hypercolumns can be rewritten as:
$$
\dot D_{k} = - D_k - E_k - \omega (\sigma(s_1) - \sigma(s_k)) 
$$
$$
\dot{E}_{k} =  \overline g (\sigma(s_1) - \sigma(s_k)) - \alpha E_{k}, \quad \overline g = \alpha g_a. 
$$
Notice that the softmax function $\sigma(s_i)$ of the vector $s_i$ as defined in \eqref{softmax_definition} is a monotone function, \ie 
	$$
	(a-b)^\top (\sigma(a) - \sigma(b)) \geq 0, \forall a, b \in \mathbb{R}^n. 
	$$
	\iftrue
	Define the following Lyapunov function and denote by $\norm{\cdot}$ the Euclidean norm:
	\begin{equation}\label{lyap_def}
	\begin{aligned}
	V = \sum\limits_{k=2}^N V_k \nonumber, \quad V_k &= \frac{1}{2} \norm{\overline g D_k + \omega E_k }^2 + \frac{1}{2}\beta\omega \norm{D_k}^2, 
	\end{aligned}
	\end{equation}
	where $\beta = \overline g + (\alpha - 1)\omega$.
	By assuming $\omega< \frac{\alpha}{1-\alpha}g_a$, we guarantee that $\beta >0$, thus the function in \eqref{lyap_def} is positive definite. We are interested in the conditions on $\omega$ such that the error dynamics is asymptotically stable, i.e. the conditions on $\omega$ such that $\dot V<0$. We have,
	\begin{gather}
		\dot V =  - \sum\limits_{k=2}^N  (\norm{\overline gD_k + \omega E_k}^2 + \nonumber\\
		+ \beta \overline g D_k^\top E_k + \beta \omega \norm{E_k}^2 + \omega \beta \norm{D_k}^2 + \omega D_k^\top E_k + \nonumber\\
		+ \omega ^2 \beta D_k^\top (\sigma(s_1) - \sigma(s_k))).
	\end{gather}
	\label{eq:Lyap_complete_derivative}
	Rearranging the terms, we obtain	
	\begin{gather}
		\dot V =  - \sum\limits_{k=2}^N( \norm{\overline gD_k + \omega E_k}^2 + \nonumber \\
		+ \omega \beta \left (\norm{E_k}^2 + \norm{D_k}^2 + \left(1 + \frac{\overline g}{\omega}\right) D_k^\top E_k\right )  + \nonumber  \\
		+ \omega ^2 \beta D_k^\top (\sigma(s_1) - \sigma(s_k))).
		\label{eq:Lyap_complete_derivative_rearranged}
	\end{gather}
	The quantity $\omega \beta \left (\norm{E_k}^2 + \norm{D_k}^2 + \left(1 + \frac{\overline g}{\omega}\right) D_k^\top E_k\right )$ is positive definite if $1 + \frac{\overline g}{\omega} <2$, i.e. $\omega > \alpha g$. Therefore, a sufficient condition for synchronization is
	\begin{equation} 
	\begin{aligned}
	\omega_{min} = g_a \alpha < \omega < \frac{\alpha}{1 - \alpha}g_a = \omega_{max}.
	\end{aligned}
	\label{eq:omega_bounds}
	\end{equation} 
\end{proof}	
\subsection{Analysis of the dynamics of one hypercolumn in the synchronized network}
Assuming that the network is synchronized, we can write the dynamics of the single module (hypercolumn) as
	\begin{equation}
	\begin{aligned}
	\dot{s}_{i1}&=-s_{i1}-a_{i1}+ \kappa o_{i1}-\frac{\kappa}{2},\\
	\tau \dot{a}_{i1}&=g_a o_{i1} - a_{i1},\\
	\dot{s}_{i2}&=-s_{i2}-a_{i2}+ \kappa o_{i2}-\frac{\kappa}{2},\\
	\tau \dot{a}_{i2}&=g_a o_{i2} - a_{i2},\\
	\end{aligned}
	\label{eq:complete_dynamics}
	\end{equation}
	where $\kappa=(N-1)\omega$.	
Recall that $o_{i1}+o_{i2}=1$. Define, $d_i = s_{i1}-s_{i2}$, and 
	\begin{equation}
	f(d_i)= o_{i1}-o_{i2}=\frac{\e^{d_i}-1}{\e^{d_i}+1}=\tanh(\frac{d_i}{2}).
	\label{soft_diff}
	\end{equation}
We now introduce the following change of variables that will help us to analyze the dynamics of each hypercolumn and earn a deeper insight into the behaviour of our system:
	\begin{equation}
	\begin{aligned}
	d_i &= s_{i1} - s_{i2}, \\
	e_i &= a_{i1} - a_{i2}. \\
	\end{aligned}
	\end{equation}
{\bf{Bifurcation Analysis}}\\[1mm]
From \eqref{eq:complete_dynamics}, the dynamics of $d_i,e_i$ obeys
\begin{equation} 
\begin{aligned}
\dot d_{i} &= - d_i - e_i + \kappa f(d_i), \\
\dot{e}_{i} &= \frac{g_a f(d_i) - e_{i}}{\tau }. \\
\end{aligned}
\label{eq:diff_dynamics_one}
\end{equation}
We start the analysis of the two dimensional system in \eqref{eq:diff_dynamics_one} by computing the equilibria $(d_i^*,e_i^*)$ and study their stability properties with the variations of the coupling parameter $\kappa$. The equilibria of the dynamics in \eqref{eq:diff_dynamics_one} satisfy the following equations 
	\begin{equation}
	(\kappa - g_a)f(d_i^*) - d_i^* = 0,
	\label{eq:eq_d_eq}
	\end{equation}
	\begin{equation}
	e_i^* = g_a f(d_i^*).
	\label{eq:eq_e_eq}
	\end{equation}
Recall from \eqref{soft_diff} that $f(d_i)= \tanh(\frac{d_i}{2})$. We can show that the system has a unique equilibrium if $\kappa < g_a+2$, while for $\kappa > g_a+2$ we have multiple equilibria. 
\begin{lemma}
The origin is the unique equilibrium for the system in \eqref{eq:diff_dynamics_one} if $\kappa < g_a+2$.
\end{lemma}
\begin{proof}
For $f(d_i)= \tanh(\frac{d_i}{2})$, we have $|f(d_i)| \leq |\frac{d_i}{2}|$. Hence if $\kappa-g_a<2$, the equality in \eqref{eq:eq_d_eq} is satisfied only for $d_i=0$.
\end{proof}
{\color{white}.}\\[2mm]
Now, assume that $\kappa <g_a +2$. The unique equilibrium is $(d_i^*,e_i^*)=(0,0)$. The computed Jacobian matrix $J_{(0,0)}$ and the characteristic polynomial are the following:	
	\begin{equation}
	J_{(0, 0)} = \begin{bmatrix} 
	-1 + \frac{\kappa}{2} & -1 \\
	\frac{g_a}{2\tau} & -\frac{1}{\tau}
	\end{bmatrix},
	\label{eq:Jac}
	\end{equation}
    \begin{equation}\label{chr}
	\rho(\lambda) = \lambda^2  -\sigma(\kappa) \lambda + \delta(\kappa),
	\end{equation}		
where $$\sigma(\kappa)=\frac{-2\tau - 2 + \tau \kappa}{2\tau},$$ $$\delta(\kappa)=\frac{-2\kappa + 2 g_a + 4}{4\tau}.$$	
 
Thus, the origin is asymptotically stable if the following two conditions hold
	\begin{align}
	\sigma(\kappa) &< 0 \quad \text{if}\quad \kappa < 2\left(1 + \frac{1}{\tau}\right),\\
	\delta(\kappa) &> 0 \quad \text{if}\quad  \kappa < g_a + 2.
	\label{hurwutz}
	\end{align}
As a result, the origin is the asymptotically stable equilibrium point (attractor) for \eqref{eq:diff_dynamics_one} if $\kappa<2(1+\frac{1}{\tau})$ holds. This condition is not particularly interesting since it means that all the units in the network converge to the same value and their output is identical. In other words, the network is not able to recall any pattern (see Fig. \ref{f2}-A). By increasing $\kappa$, the attractor of the dynamical system \eqref{eq:diff_dynamics_one} changes. 
In fact, at the critical value $\kappa^* = 2 \left( 1 + \frac{1}{\tau} \right)$ a {\em supercritical Hopf} bifurcation \cite{Guckenheimer1983} occurs. That is, a unique stable limit cycle bifurcates from the origin. This result is now formally presented below.
\begin{proposition}\label{pr}
For the system in \eqref{eq:diff_dynamics_one}, a unique stable limit cycle bifurcates from the fixed point $(d_i,e_i)=(0,0)$ into the region $\kappa > 2(1+\frac{1}{\tau})$ if $g_a >\frac{2}{\tau}$ holds.
\end{proposition}
\begin{proof}
To prove the above statement, we show that all of the conditions of \cite[Theorem~3.4.2]{Guckenheimer1983} are satisfied. Denote the eigenvalues of $J_{(0,0)}$ in \eqref{eq:Jac} as a function of the parameter $\kappa$ by $\lambda_{1,2}(\kappa)$. From \eqref{chr}, We have
	\begin{equation*}
	\begin{aligned}
	\lambda_{1,2}(\kappa) = \frac{1}{2}\left(\sigma(\kappa) \pm \sqrt{\sigma(\kappa)^2-4\delta(\kappa)} \right)
	\end{aligned}.
	\end{equation*}
At the critical value $\kappa^*= 2 \left(1 + \frac{1}{\tau} \right)$, the following conditions should be satisfied:
\begin{enumerate}
\item 
The Jacobian matrix, $J_{(0,0)}$ in \eqref{eq:Jac}, has a conjugate pair of imaginary eigenvalues. That is,
		\begin{align}
		\sigma(\kappa^*) =0,\quad
		\delta(\kappa^*) >0. \label{eq:non_hyp_cond}
		\end{align}
The above condition is satisfied if $g_a>\frac{2}{\tau}$ holds.
\item 
Eigenvalues of $J_{(0,0)}$ in \eqref{eq:Jac} vary smoothly by $\kappa$, \ie
\begin{align}
\left. \frac{\partial \sigma(\kappa)}{\partial \kappa}\right|_{\kappa=\kappa^*}= \frac{1}{2} \neq 0.
\end{align}
\item 
The last condition is to prove bifurcation of a stable limit cycle. We first rewrite the system in \eqref{eq:diff_dynamics_one} in the follwoing form
	\begin{equation}
	 \begin{bmatrix} 
	{\dot d}_i \\
	{\dot e}_i 
	\end{bmatrix} = \underbrace{\begin{bmatrix} 
	\frac{1}{\tau} & -1 \\
	\frac{g_a}{2\tau} & -\frac{1}{\tau}
	\end{bmatrix}}_{A}\begin{bmatrix} 
	{\dot d}_i \\
	{\dot e}_i 
	\end{bmatrix} + \begin{bmatrix} F_1(d_i) \\F_2(d_i) \end{bmatrix},
	\label{eq:r2}
	\end{equation}
with $F_1(d_i)=(-1-\frac{1}{\tau}) d_i + \kappa \tanh(\frac{d_i}{2})$ and $F_2(d_i)= -\frac{g_a}{2 \tau} d_i + \frac{g_a}{\tau}\tanh(\frac{d_i}{2})$. Denote the eigenvalues of the matrix $A$ in \eqref{eq:r2} by $\pm \beta i$, $\beta>0$. Notice that the latter are also the eigenvalues of $J_{(0,0)}$ at ${\kappa=\kappa^*}$. We now consider a coordinate transformation
$\begin{bmatrix} 
{d}_i \\{e}_i \end{bmatrix} = E \begin{bmatrix} {u}_i \\{v}_i \end{bmatrix}$ such that $$E=\begin{bmatrix} 0 & 1 \\-\beta & \frac{1}{\tau}\end{bmatrix}, \begin{bmatrix} 
{u}_i \\{v}_i \end{bmatrix} = E^{-1}\begin{bmatrix} {d}_i \\{e}_i \end{bmatrix}= \begin{bmatrix} \frac{d_i}{\tau \beta}-\frac{e_i}{\beta} \\{d}_i 
\end{bmatrix}.$$
Notice that in the above $d_i=v_i$, hence $F_1(d_i)=F_1(v_i),F_2(d_i)=F_2(v_i)$. We write the dynamics of \eqref{eq:diff_dynamics_one} in $(u_i,v_i)$ coordinate which gives
	\begin{equation}
	 \begin{bmatrix} 
	{\dot u}_i \\
	{\dot v}_i 
	\end{bmatrix} = \begin{bmatrix} 
	0 & -\beta \\
	\beta & 0
	\end{bmatrix}\begin{bmatrix} 
	{u}_i \\
	{v}_i 
	\end{bmatrix} + E^{-1} \begin{bmatrix} F_1(v_i) \\F_2(v_i) \end{bmatrix}.
	\label{eq:r1}
	\end{equation}
We now can calculate the cubic coefficient of the Taylor series of degree 3 for the above system based on the formula (3.4.11) of \cite{Guckenheimer1983}. The computation gives a negative value, $-\frac{\kappa^*}{64}$, which ends the proof.
\end{enumerate}
\end{proof}

Proposition 1 gives a sufficient condition for synchronization of the network in \eqref{s_dynaics_recall-n}-\eqref{a_dynaics_recall-n} and Proposition 2 presents a necessary condition such that the synchronous state is a stable limit cycle, \ie\  a stored pattern is recalled. We now combine the two results below.
\begin{corollary}\label{cor1}
For the network dynamical system in \eqref{s_dynaics_recall-n}- \eqref{a_dynaics_recall-n}, $\frac{g_a}{\tau} < \omega < \frac{g_a}{\tau-1}$ is a sufficient condition for the network to converge to the set ${\S}=\{s_i \in \R, a_i \in \R: s_i = s_k; a_i=a_k\  \forall i, k \in \{1,\ldots,N\}\}$. Moreover if $\omega < \frac{g_a + 2}{N-1}$ holds, then the origin is the unique equilibrium for the relative state, $(d_i,e_i)$ with $d_i=s_{i,1}-s_{i,2}, e_i=a_{i1} - a_{i2}, \forall i$. Furthermore, the two conditions $g_a >\frac{2}{\tau}$, and $\omega > \frac{2}{N-1}(1+\frac{1}{\tau})$ are necessary for achieving a stable limit cycle as the attractor of the relative state $(d_i,e_i)$ of each $s_i$ on the set $\S$.
\end{corollary}
By increasing $\kappa$, the attractor of the dynamics in \eqref{eq:diff_dynamics_one} changes. In what follows, we present a numerical example explaining these changes. The numerical example is also presented in Figure \ref{f2} (plotted using MATCONT \cite{dhooge2003matcont} and Python simulations). 

\begin{example}[Numerical bifurcation analysis]
For the system (representing one hypercolumn) in \eqref{eq:complete_dynamics}, set $\tau=2,\ g_a=10$. The relative states $(d_i,e_i)$ with the dynamics in \eqref{eq:diff_dynamics_one} converge to
\begin{itemize}
\item a stable equilibrium point for $\kappa<2(1+\frac{1}{\tau})=3$. 
In this case, no pattern is recalled, \ie\  both minicolumns reach the same level of activation, Fig. (\ref{f2}-A),
\item a stable limit cycle with $2(1+\frac{1}{\tau}) < \kappa <\kappa^{oo}$, $g_a + 2=12 < \kappa^{oo} \approx 13.11$. In this case the two minicolumns are activated in turn, that is the two stored patterns are reactivated in turn, Fig. (\ref{f2}-B),
\item a limit cycle and two locally stable equilibria points with $13.11 \lessapprox \kappa \lessapprox 13.24605$. In this case, both oscillatory behaviour (cyclic activation) and fixed point behaviours are present. Depending on the initial conditions, the system converges to one of the three attractors, Fig. (\ref{f2}-C, \ref{f2}-D), 
\item two locally stable equilibria with $\kappa \gtrapprox 13.24605$. In this case, the system does not show any oscillatory behaviour but it locally converges to one of the two stored patterns, \ie\  one of the two minicolumns locally reaches and stays at a higher level of activation, Fig. (\ref{f2}-E).
\end{itemize}
\end{example}

\pagebreak
\begingroup
\let\clearpage\relax 
\onecolumn 
\begin{figure}[!h]
\centering
\includegraphics[width=\textwidth]{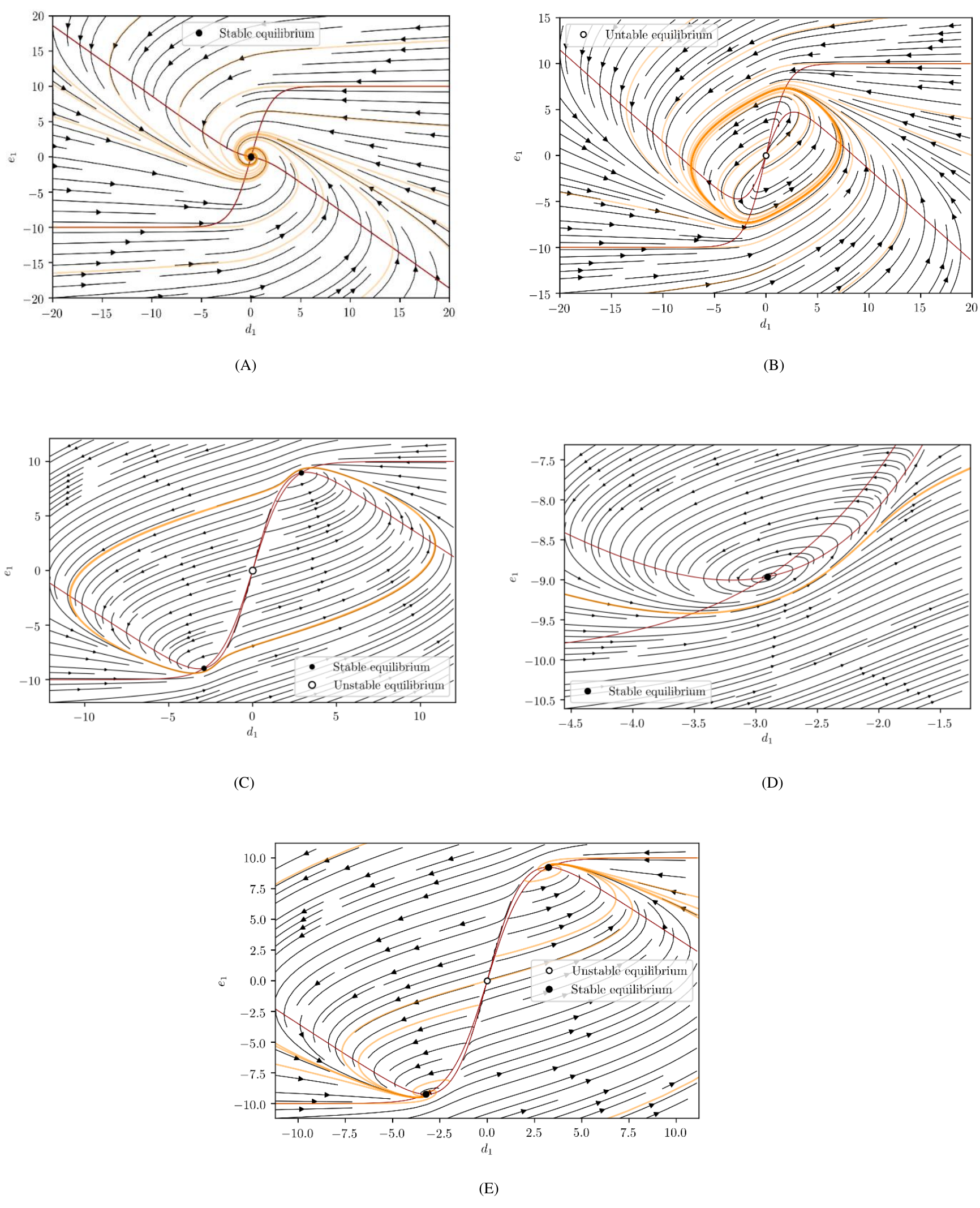}
\caption{\label{f2}Phase portrait of system in \eqref{eq:diff_dynamics_one} for $g_a=10,\ \tau=2$. The nullclines are shown in red; the streamlines are shown in black. Several trajectories of the system with different initial conditions are shown in orange. The plots show changes of the attractor for the dynamics in \eqref{eq:diff_dynamics_one} based on the increase of $\kappa$. In plots C, D where both fixed points and limit cycles (LC) exist, the LCs are plotted in Orange. Plot (D) shows a magnification of Plot (C).}
\end{figure}
\endgroup
\begingroup
\twocolumn
\section{Simulation results}\label{sec:sim}
We consider a network of $N=12$ hypercolumns, each composed of two minicolumns. We set, $g_a=97$, $\tau=54$ (implemented as $\frac{\tau_a}{\tau_m}=\frac{27}{0.05}$, see \cite{lansner2013reactivation}). Based on Corollary 1, considering both the sufficient condition for achieving synchronization and the necessary condition such that each hypercolumn attractor is a limit cycle, we calculate $\omega=1.8$, hence $\kappa=19.8$. Figure \ref{f3} shows the relative state of positively correlated minicolumns. For clarity of presentation, the relative states with respect to hypercolumn 1, \ie\  variables $D_{1,\ell}=s_{1,1}-s_{\ell,1}$ and $E_{1,\ell}=a_{1,1}-a_{\ell,1}$ with $\ell \in \{2,..,N\}$, are plotted. As shown the relative states of positively correlated minicolumns in all hypercolumns converge to zero. The latter implies synchronization. Figures \ref{f4} and \ref{f5} show the dynamics of one hypercolumn. As shown in Figure \ref{f4}, the two minicolumns oscillate in turn. Figure \ref{f5} shows the output of the two minicolumns and the level of their corresponding adaptation variables.  
 As shown, the role of adaptation variable is to inhibit the minicolumn which is most active and reactivate it in turn.  
\begin{figure}[h]
\centering
\includegraphics[width=0.4\textwidth]{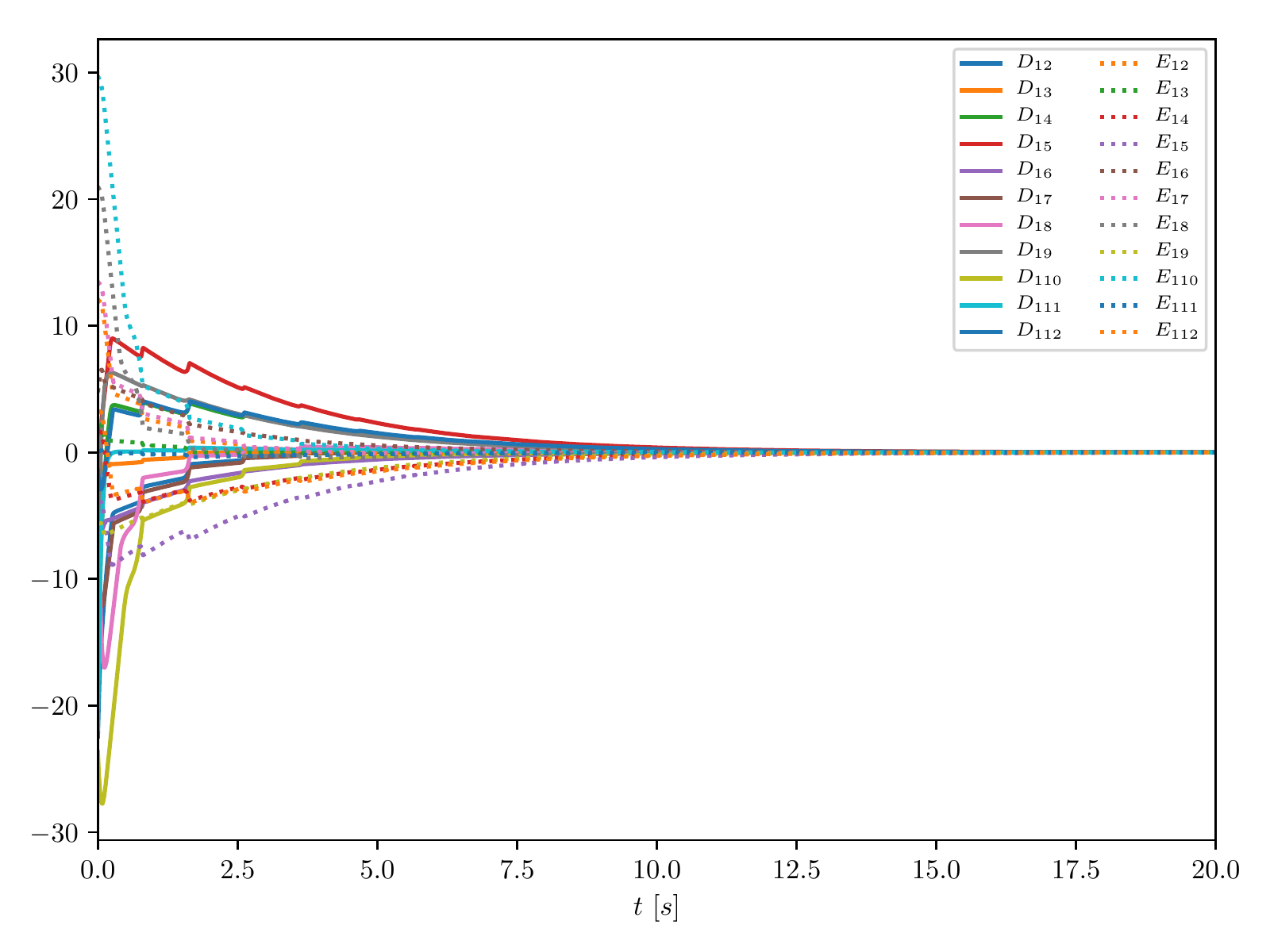}
\caption{\label{f3} The plot of the relative states of positively correlated minicolumns with respect to hypercolumn 1.}
\end{figure}
\begin{figure}[h]
\centering
\includegraphics[width=0.38\textwidth]{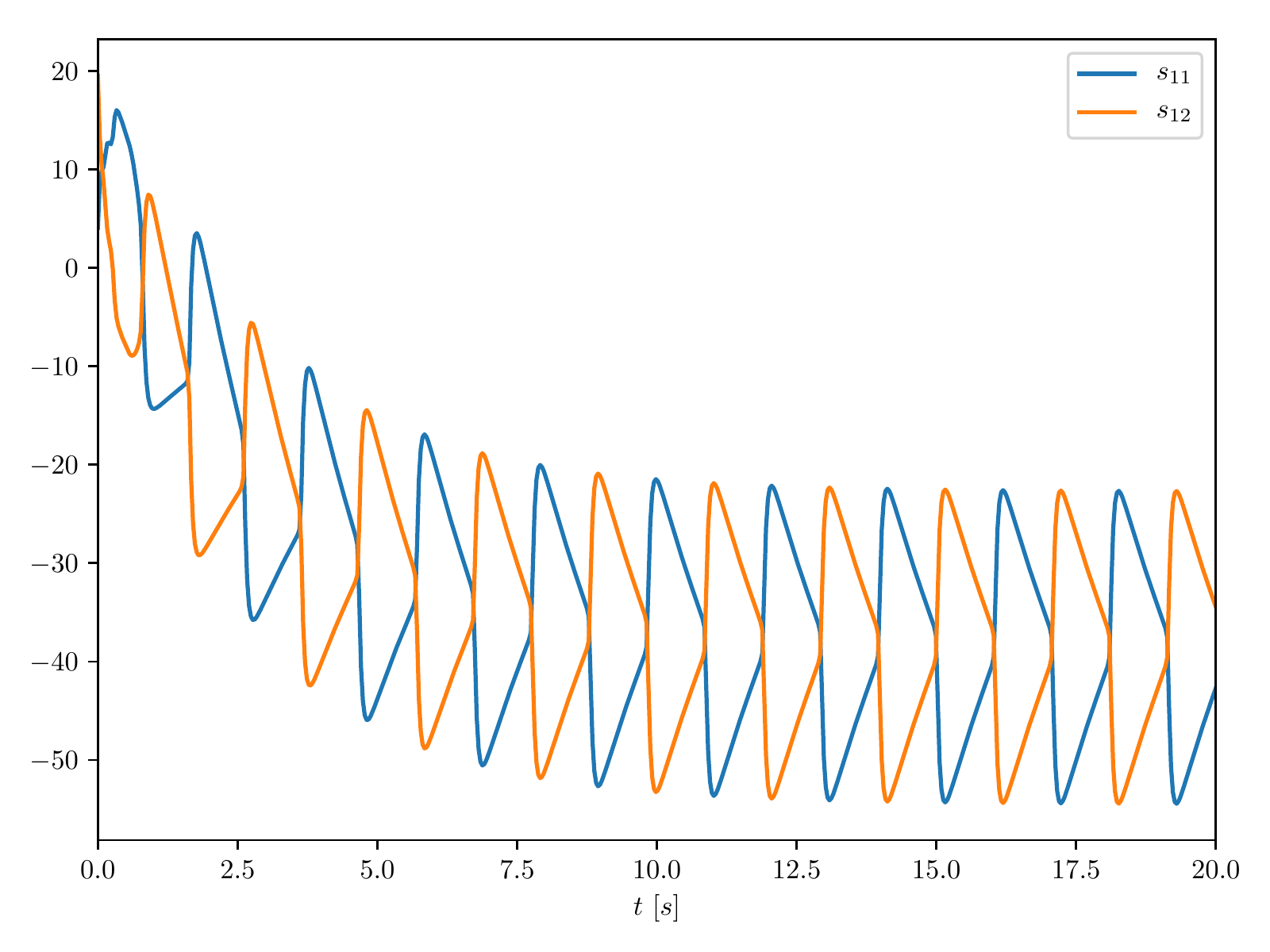}
\caption{\label{f4} The state of hypercolumn 1. Variables $s_{11}$ and $s_{12}$ achieves alternative activation and deactivation.}
\end{figure}
\begin{figure}[h]
\centering
\includegraphics[width=0.38\textwidth]{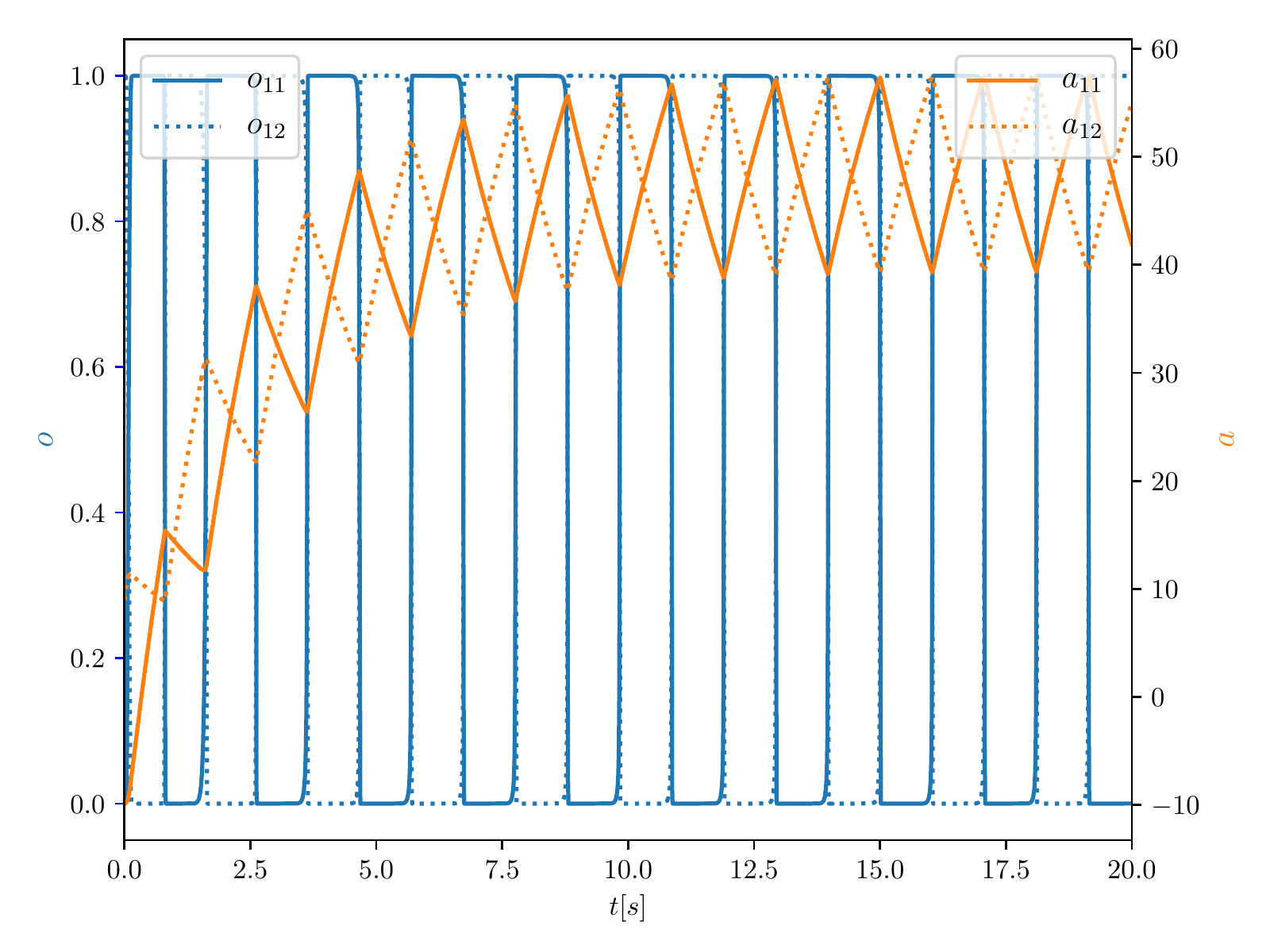}
\caption{\label{f5} Output and adaptation variables of a single hypercolumn. The adaptation variables are plotted in orange, the output of the units in blue. The inhibition of one unit allows the other unit to activate.}
\end{figure}
\section{Conclusions}\label{sec:conclude} 
This paper has studied the free recall dynamics of a multi-item working memory network modeled as a simplified modular attractor neural network. The network consisted of $N$ modules each one composed of two units. We analyzed the free recall dynamics assuming a constant, and homogeneous coupling between the network modules. A sufficient condition for synchronization of the network was obtained. Furthermore, assuming a synchronized network, the behavior of one module was analyzed using bifurcation analysis tools. Based on this analysis, a necessary coupling condition was given for achieving a stable limit cycle behavior for each module in the synchronized network. The latter implies the memory is recalling a stored pattern.
\bibliographystyle{plain}
\bibliography{biblio}
\endgroup
\end{document}